\documentclass[sigconf]{acmart}

\usepackage{xcolor}
\usepackage{enumitem}
\usepackage{balance}
\usepackage{multirow}
\usepackage{dsfont}

\definecolor{boxgrey}{HTML}{F3F3F3}

\newlist{enumerateinline}{enumerate*}{1}
\setlist[enumerateinline]{
    label=(\roman*)
    ,productjoin={{, }}
    ,productjoin*={{, and }}
    ,after=\unskip{. }
}
\newlist{enumerateinlineempty}{enumerate*}{1}
\setlist[enumerateinlineempty]{
    label=(\roman*)
    ,productjoin={{; }}
    ,productjoin*={{ and }}
}

\AtBeginDocument{
  \providecommand\BibTeX{{
    \normalfont B\kern-0.5em{\scshape i\kern-0.25em b}\kern-0.8em\TeX}}}
    
\copyrightyear{2022}
\acmYear{2022}
\setcopyright{acmcopyright}\acmConference[SIGIR '22]{Proceedings of the 45th International ACM SIGIR Conference on Research and Development in Information Retrieval}{July 11--15, 2022}{Madrid, Spain}
\acmBooktitle{Proceedings of the 45th International ACM SIGIR Conference on Research and Development in Information Retrieval (SIGIR '22), July 11--15, 2022, Madrid, Spain}
\acmPrice{15.00}
\acmDOI{10.1145/3477495.3531760}
\acmISBN{978-1-4503-8732-3/22/07}

\acmSubmissionID{1370}
 
\begin{document}
\fancyhead{}

\title[Regulating Group Exposure for Item Providers in Recommendation]{Regulating Group Exposure for Item Providers \\ in Recommendation}

\author{Mirko Marras}
\orcid{0000-0003-1989-6057}
\affiliation{
  \institution{University of Cagliari}
  \city{Cagliari}
  \country{Italy}
}
\email{mirko.marras@acm.org}

\author{Ludovico Boratto}
\orcid{0000-0002-6053-3015}
\affiliation{
  \institution{University of Cagliari}
  \city{Cagliari}
  \country{Italy}
}
\email{ludovico.boratto@acm.org}

\author{Guilherme Ramos}
\orcid{0000-0001-6104-8444}
\affiliation{
  \institution{LASIGE, Faculdade de Ci\^encias, University of Lisbon}
  \city{Lisbon}
  \country{Portugal}
}
\email{ghramos@fc.ul.pt}

\author{Gianni Fenu}
\orcid{0000-0003-4668-2476}
\affiliation{
  \institution{University of Cagliari}
  \city{Cagliari}
  \country{Italy}
}
\email{fenu@unica.it}

\renewcommand{\shortauthors}{Marras, et al.}

\begin{abstract}
Engaging all content providers, including newcomers or minority demographic groups, is crucial for online platforms to keep growing and working. Hence, while building recommendation services, the interests of those providers should be valued. 
In this paper, we consider providers as grouped based on a common characteristic in settings in which certain provider groups have low representation of items in the catalog and, thus, in the user interactions. Then, we envision a scenario wherein platform owners seek to control the degree of exposure to such groups in the recommendation process. To support this scenario, we rely on disparate exposure measures that characterize the gap between the share of recommendations given to groups and the target level of exposure pursued by the platform owners. We then propose a re-ranking procedure that ensures desired levels of exposure are met. Experiments show that, while supporting certain groups of providers by rendering them with the target exposure, beyond-accuracy objectives experience significant gains with negligible impact in recommendation utility. 
\end{abstract}



\maketitle

\section{Introduction}

\noindent {\bf Motivation.} Online platforms are facilitating interactions among multiple parties, usually consumers and providers. These environments often rely on recommender systems that use predicted relevance to match consumers to providers~\cite{RicciRS15}. Optimizing recommendations only for the former party has been seen, for years, as the final goal. Nevertheless, recommender systems are multi-sided and should carefully consider the interests and needs of providers other than those of consumers~\cite{Burke17,marras2021equality}. One aspect receiving attention is that certain groups of providers often have little representation in the past interactions and, hence, low exposure in the recommendations. For instance, movies directed by women in ML-1M~\cite{harper2016movielens} are a minority of $10\%$ in the catalog, and their representation drops to $7\%$ in the interactions. With a pairwise recommender~\cite{song2018neural}, movies from female directors get exposure of 5\%. Hence, the recommender system may propagate providers' disparate exposure. 

\vspace{2mm} \noindent {\bf Problem Statement.} Mitigating these disparities should not only be pursued by the law if they involve legally-protected classes of providers but also could come from the platform’s business model. For instance, a platform such as Udemy may wish to ensure that courses of new teachers get a given amount of recommendations, though they have fewer learners than courses from established teachers. Similarly, platforms such as Kickstarter may desire to support new projects or projects in sensitive domains and provide them with more exposure that may enable backers to support such projects. Ensuring a certain degree of exposure to all provider groups relevant for the platform may also encourage vivid dynamics, leading to beyond-accuracy benefits for the platform and its stakeholders (e.g., novelty and diversity)~\cite{kaminskas2016diversity}. Supporting provider groups, especially minorities, is hence of primary importance. 

\vspace{2mm} \noindent {\bf Open Issues.} In the literature, there exist several recommendation procedures that re-rank the outputs of the original recommender or ranking system to meet certain diversification-related ranking objectives. These procedures do not often account for the position bias in the ranked outputs and just focus on their visibility (i.e., the fraction of a given provider group's items present in the top ranking prefix~\cite{DBLP:conf/cikm/ZehlikeB0HMB17,DBLP:conf/ssdbm/YangS17}). Other procedures, often proposed under a fairness framework, aim to ensure equitable exposure to groups of providers based on equity-driven objectives but do not avoid situations where items of minority groups get ranked at the bottom of the top ranking prefix~\cite{DBLP:conf/ecir/GomezBS21,DBLP:conf/www/Zehlike020,DBLP:conf/kdd/SinghJ18,DBLP:conf/sigir/BiegaGW18}. 

\vspace{2mm} \noindent {\bf State of the Art.} Our approach in this paper differs from the prior work in three major ways. First, compared to diversification-related fair ranking works, our approach accounts for the position bias in the ranked outputs~\cite{DBLP:conf/cikm/ZehlikeB0HMB17,DBLP:conf/ssdbm/YangS17}. Second, in contrast to approaches similar to~\cite{DBLP:conf/sigir/BiegaGW18}, our approach does not require rank-aware quality metrics that are inherently not available to the recommender system at the moment of optimization. Moreover, our approach is based on maximum marginal relevance, with a combinatorial space linear in the number of items. Other methods often rely on an integer linear program that operates on a combinatorial space quadratic in the number of items, requiring assumptions on the pre-filtering selection of items. Third, compared to~\cite{DBLP:conf/ecir/GomezBS21}, our approach can achieve the desired balancing goal on the ranking of each user without any knowledge of all other users, and is parametrized on the target exposure the platform owners seek for those provider groups, according to their policies. Our approach can be applied to the output of any recommender system, in contrast to in-processing approaches~\cite{DBLP:conf/www/Zehlike020}. 

\vspace{2mm} \noindent {\bf Our Contributions.} In this paper, we consider a scenario where providers are grouped based on a common characteristic and certain provider groups have a low representation of items in the catalog and, then, in the recommendations. We then envision that platform owners seek to guarantee a certain degree of exposure to all provider groups, including those minorities, while recommending. Under this scenario, we provide the following contributions:

\begin{itemize}[leftmargin=*]
\item We propose a post-processing approach to ensure a given degree of exposure to providers' groups by optimizing the maximum marginal relevance for the Hellinger distance between the target and the achieved exposure distributions.
\item We evaluate our approach against state-of-the-art baselines to assess how supporting all provider groups under a given recommendation policy pursued by the platform owners impacts accuracy and beyond-accuracy aspects under two datasets.
\end{itemize}

\section{Providers' Exposure Framework}
\label{sec:minority-support}
In this section, we define the notation, group disparity metrics, and a new re-ranking approach to adjust exposure among providers. 

\subsection{Preliminaries}
Given a set of users $U$ and a set of items $I$, we assume that users have expressed their interest in a subset of items in $I$. The collected feedback from observed user-item interactions can be abstracted to a set of ($u$, $i$) pairs, implicitly obtained from natural user activity, which we shortly denote by $R_{u,i}$. We denote the user-item feedback matrix $R \in \mathbb{R}^{M*N}$ as by $R_{u,i}=1$ to indicate user $u$ interacted with item $i$, and $R_{u,i}=0$ otherwise. Furthermore, we denote the set of items user $u\in U$ interacted with by $I_u=\{i\in I\,:\,R_{u,i}=1\}$.  

We consider an attribute that characterizes providers (e.g., gender) and can assume a value from a set $A$ (e.g., $A=\{male,female\}$). We define a function $\mathcal{F}: I \longrightarrow A$ that returns the value of the attribute for a provider of a given item. For each item $i$, if the provider of $i$ has the attribute $a$, we consider that $\mathcal{F}(i) =a$. The set of items whose providers have attribute $a\in A$ is denoted by $I^a=\{i\in I\,:\,\mathcal{F}(i)=a\}$, and the set of items user $u$ interacted with and come from a provider with attribute $a$ is denoted by $I_u^a=I_u\cap I^a$. 

We assume that each user $u \in U$ and item $i \in I$ is internally represented by a $D$-sized numerical vector from a user-vector matrix $W$ and an item-vector matrix $X$, respectively. The recommender system's task is to optimize $\theta = (W,X)$ for predicting unobserved user-item relevance. It can be abstracted as learning $\widetilde{R}_{u,i} = f_{\theta}(u,i)$, where $\widetilde{R}_{u,i}$ denotes the predicted relevance, $\theta$ denotes learnt user and item matrices, and $f$ denotes the function predicting the relevance of item $i$ for user $u$. Given a user $u$, items $i \in I \setminus I_u$ are ranked by decreasing $\widetilde{R}_{u,i}$, and top-$k$, with $k\in\mathbb N$ and $k>0$, items are recommended. For conciseness, our study will focus on $k=10$. Finally, let us define a function $\mathcal{G}$, such that $\mathcal{G}(u,i|k)=p$ if item $i$ is recommended to user $u$ at position $p$. 

\subsection{Group Disparity Formulation}\label{sec:measure_cal} 
We formalize disparate exposure as the distance between the degree of exposure received by providers' groups in recommendations and the degree of exposure targeted for each of them by platform owners, according to a given recommendation policy. The higher the dissimilarity is, the higher the group's disparate exposure. 

We resorted to such a global notion of disparate exposure locally on each ranked list, so that it can be optimized via a re-ranking function. For the ranked list of a user $u$, we assume that platform owners seek to ensure a targeted degree of exposure $p$ for items $i \in I$ whose providers have the attribute $a$, formalized as: 

\begin{equation}
    p_u(a)\in[0,1]\,\text{ and }\,\sum_{a\in A}p_u(a)=1
\end{equation}

\noindent When a recommender system is used to suggest the top-$k$ items to user $u \in U$, we define the degree of exposure $q$ achieved in the recommendations by items $i \in I$ whose providers have attribute $a$:
\begin{equation*}
\vspace{-1mm}
	q_u(a|k)=\frac{\sum_{p=1}^{k}  1 / log_2(p+1) \cdot \mathds{1}(\mathcal{F}(i_p) = a)}{\sum_{p=1}^{k}  1 / log_2(p+1)}\,\text{ and }\,\sum_{a\in A}q_u(a|k)=1, 
\end{equation*}
where $i_p$ is the item recommended to $u$ at position $p$; $q_u(a|k)$ ranges in $[0, 1]$, being 0 when there is no exposure for items of the group $I_a$ in the top-$k$. Positive values indicate the share of exposure for items of the group $I_a$ at the prefix $k$ of the ranking. 

For the ranked list of a user $u$, the degree of exposure targeted by platform owners for each provider group is pursued if the vectors $p_u(k)$ and $q_u(k)$ are statistically indistinguishable. This may be sorted out by doing a statistical test of the hypothesis, such as the \textit{Kolmogorov–Smirnov} test, one of the most computationally efficient asymptotic tests. Nevertheless, such statistical tests lead to strictly combinatorial optimization problems. Hence, to support the exposure goals targeted by platform owners, we approximately and directly compare the vectors $p_u(k)$ and $q_u(k)$ to quantify the similarity between them. Though the \textit{Kullback–Leibler} divergence is commonly used to compare such two distributions, the fact that it is non-symmetric and unbounded might lead to low interpretability and computational instability. Another popular solution that we leave for future work is the L1 distance, which however is less intuitive and more likely to give high distance values. For these reasons, as a support metric in our study, we use the \textit{Hellinger} distance, which is both symmetric and bounded in the range [0,1]. Specifically, the support metric $H(p,q|k)$ we consider is defined as: 

\begin{equation}
\begin{split}
H(p_u, q_u|k) = \frac{1}{\sqrt 2}\sqrt{\sum_{a\in A}\left(\sqrt{p_u(a)}-\sqrt{ q_u(a|k)}\right)^2}
\end{split}
\label{eq:hellinger}
\end{equation}

\noindent where $p_u = (p_u(a_1),\ldots,p_u(a_m))$ and $q_u = (q_u^e(a_1|k),\ldots, q_u^e(a_i|k))$, within the attribute $A=\{a_1,\ldots, a_m\}$. Hence, $H(p_u,q_u|k)=0$ if $p_u(k)$ and $q_u(k)$ are \emph{perfectly balanced}, meaning that the exposure levels pursued by platform owners are met. Conversely, the maximum distance $1$ is achieved when $p_u(k)$ assigns 0 to every attribute that $q_u(k)$ assigns a positive value (or vice versa), so that the distributions are \emph{completely unbalanced}. In the latter situation, the recommender suggests items by rendering certain degrees of group exposure in the opposite direction to the platform owners' goals.   

\subsection{Disparity Control Procedure}\label{sec:cal_frame} 
To achieve the degree of exposure pursued by platform owners for each provider group, we introduce a recommendation procedure that seeks to minimize the support metric provided in Eq.~\eqref{eq:hellinger}. Since, in general, it is hard to plug in the balancing phase inside a recommender system, we propose to balance the output obtained by a recommender system through a re-ranking of the recommended list that it returns, a common practice in the recommender system literature~\cite{RicciRS15}. 
For each user $u\in U$, our goal is to determine an optimal set $\mathcal I^\ast$ of $k$ items to be recommended to $u$, so that the targeted degree of exposure $p$ is met. To this end, we adopt a \emph{maximum marginal relevance}~\cite{carbonell1998use} approach, with Eq.~\eqref{eq:hellinger} as the support metric. 
The set $\mathcal I^\ast$ is obtained by solving the optimization problem: 

\begin{equation}\label{eq:opt_prob}
    \mathcal I^\ast(u|k) = \mathop{\text{argmax}}_{\mathcal I\subset I,|\mathcal I|=k}\,\lambda \sum_{i\in \mathcal I}\widetilde R_{ui}-(1-\lambda)\,H^2 (p_u,{q_u}_{|\mathcal I})
\end{equation}

\noindent where ${q_u}_{|\mathcal I}$ is $q_u$ when the top-$k$ list includes items $\mathcal I$, and $\lambda\in[0,1]$ expresses the trade-off between accuracy and disparate exposure. With $\lambda=1$, we simply yield the output of the recommender system, not taking disparate exposure into account. Conversely, with $\lambda =0$, the output of the recommender system is discarded, and we only focus on controlling disparate exposure. This combinatorial maximization problem may be efficiently approximated with a greedy approach with $(1-1/e)$ optimality guarantees, if the objective function of the maximization is submodular~\cite{nemhauser1978analysis}. 

\begin{proposition}

    The objective function to be maximized in Eq.~\eqref{eq:opt_prob}, i.e., $f(\mathcal I)=\lambda\textstyle  \sum_{i\in\mathcal I}\widetilde R_{ui}-(1-\lambda)\,H^2(p_u,{q_u}_{|\mathcal I})$, is submodular. 

\end{proposition}

\begin{proof}
     First, we present the submodularity proof of the relevance part of the objective (left term in the equation) for the elements of the set $\mathcal I$. Since $\widetilde R_{ui} > 0$, it follows that $\sum_{i\in\mathcal I}\widetilde R_{ui}$ is a modular function and, hence, submodular. Second, we present the submodularity proof of the disparity objective for the elements of the set $\mathcal I$ over the set $\mathcal A$. To demonstrate it, we notice that: 
    
    \begin{equation*}
    \resizebox{0.475\textwidth}{!}{%
        $
        \begin{split}
        H^2(p_u,q_{u_{|\mathcal I}})&=
        \left(\frac{1}{\sqrt 2}\sqrt{\sum_{a\in A}\left(\sqrt{p_u(a)}-\sqrt{ q_{u_{|\mathcal I}}(a)}\right)^2}\right)^2\\
        & = \frac{1}{2}\left(\sum_{a\in A} p_u(a)+q_{u_{|\mathcal I}}(a)-2\sqrt{p_u(a)q_{u_{|\mathcal I}}(a)}\right)
        \\ &
        = \frac{1}{2}\left(1+1-2\sum_{a\in A}\sqrt{p_u(a)q_{u_{|\mathcal I}}(a)}\right)
        =1-\sum_{a\in A}\sqrt{p_u(a)q_{u_{|\mathcal I}}(a)}\\
    \end{split}
    $
    }
    \end{equation*}
    
    \noindent Given the above result and replaced it into Eq.~\eqref{eq:opt_prob}, we obtain: 
    
    \begin{equation*}
        \resizebox{0.475\textwidth}{!}{%
        $
        \displaystyle\mathop{\text{argmax}}_{\mathcal I\subset I,|\mathcal I|=k}f(\mathcal I) = \mathop{\text{argmax}}_{\mathcal I\subset I,|\mathcal I|=k}\,\lambda \left(\sum_{i\in\mathcal I}\widetilde R_{ui}\right)+(1-\lambda)\left(\sum_{a\in A}\sqrt{p_u(a){q_u}_{|\mathcal I}(a)} - 1\right)
        $
        }
    \end{equation*}

    \noindent Since $\sqrt{p_u(a){q_u}_{|\mathcal I}(a)} > 0$, $\displaystyle \sum_{a\in A}\sqrt{p_u(a){q_u}_{|\mathcal I}(a)} - 1$ is modular\footnote{If a function $f$ is submodular, adding or subtracting a constant to $f$ does not change the submodularity property.}.  
    
     \noindent Therefore, the sum of modular functions is a modular function, and a modular function is also submodular, as we aimed to show. 
\end{proof}

This greedy approach yields an ordered list of items, and the resulting list at each step is $(1 - 1/e)$ optimal among the lists of equal size. This property fits with the real world, where users may initially see only the first $k < N$ recommendations, and the remaining items may become visible after scrolling. Our approach allows us also to control the exposure of multiple provider groups in the ranked lists, and it does not pose any constraint on the size of $A$. 

\subsection{Minority Recommendation Policies}

The share of recommendations given to provider groups in terms of exposure might not only be controlled by platform owners according to the law if those minorities involve legally-protected classes of providers but could also depend on the platform’s business model. Our study in this paper focuses on five recommendation policies that could be pursued by platform owners while recommending. Please, note that our disparity control approach is completely agnostic to the underlying recommendation policy. Therefore, it could be extended to any other policy that relies on a distribution of the share of recommendations. Specifically:    

\begin{itemize}[noitemsep,topsep=0pt,leftmargin=*]  
\item The \texttt{Cat} policy aims to ensure that a provider group has an exposure proportional to its representation in the catalog, i.e., $\lambda < 1$ and $p_u(a) = |I^a| / |I|$ in Eq.~\eqref{eq:opt_prob}. This policy follows a distributive norm based on equity among providers' groups~\cite{DBLP:book/Walster1973new}.   
\item The \texttt{Int} policy aims to ensure that each provider group has an exposure proportional to its representation in the interactions, i.e., $\lambda < 1$ and $p_u(a) = \sum_{u,i} R_{u,i}, \forall i \in I^a / \sum_{u,i} R_{u,i}$ in Eq.~\eqref{eq:opt_prob}. This policy aims to ensure that no distortion in recommendations is added with respect to the degree of interaction with each group.  
\item The \texttt{Par} policy aims to ensure that provider groups have the same degree of exposure among each other, i.e., $\lambda < 1$ and $p_u(a) = 1 / |A|$ in Eq.~\eqref{eq:opt_prob}. This policy follows a distributive norm based on statistical parity among providers' groups~\cite{zemel2013learning}.    
\item The \texttt{Per} policy aims to ensure that each provider group has an exposure proportional to the its representation on the profile of the current user, i.e., $\lambda < 1$ and $p_u(a) = |I^a_u| / |I_u|$ in Eq.~\eqref{eq:opt_prob}. This policy subsumes the \texttt{Int} policy, but it calibrates the share of recommendations for a user according to the individual user's preferences, not to the global degree of interaction with a group.
\end{itemize}

\section{Experimental Evaluation}
\label{sec:experimental-evaluation}
In this section, we aim to answer two research questions: 

\begin{enumerate}[label=\textbf{RQ\arabic*},leftmargin=10mm]
\item \emph{Can our re-ranking approach achieve a better trade-off between recommendation utility and disparate provider group exposure, compared to state-of-the-art baselines?} 
\item \emph{Does our re-ranking approach provide benefits pertaining to relevant beyond-accuracy objectives for the platform, such as novelty and diversity?} 
\end{enumerate}

\noindent \textbf{Data}. Our case study in this paper assumes that providers are grouped based on their gender, so $A=\{ \text{female}, \text{male}\}$. Hence, we used two public datasets that, to the best of our knowledge, are among the few including providers' gender\footnote{We used the providers' sensitive attribute labels provided by \cite{boratto2021interplay}.}. \textit{ML-1M}~\cite{harper2016movielens} includes $1M$ interactions performed by $6K$ users towards $3K$ movies. Considering movie directors as providers, the representation of movies from female directors is $10\%$ in the catalog and $7\%$ in the interactions. \textit{COCO}~\cite{dessi2018coco} includes $37K$ learners who produced $600K$ interactions with $30K$ online courses. Considering course teachers as providers, the representation of courses from female teachers is around $19\%$ in the catalog and $12\%$ in the interactions. Due to their representation, we considered female directors and teachers as minorities. 

\vspace{2mm} \noindent \textbf{Recommender Details}. We investigated the impact of our approach on a pairwise recommendation algorithm~\cite{song2018neural}. Pairwise learning is the foundation of many cutting-edge personalized ranking algorithms~\cite{zhang2019deep}. For each dataset, we applied a temporal train-test split, with the most recent $20\%$ of interactions for each user in the test set and the rest in the train set. Matrices were initialized with uniformly-distributed values in $[0, 1]$, and the size of the user and item latent factors was set to $16$. For $20$ epochs, the model is served with $1024$-sized batches. For each user, we created $10$ triplets per observed item; the unobserved item was selected randomly for each triplet. The optimizer was Adam with a learning rate of $0.001$.

\vspace{2mm}  \noindent \textbf{Baselines}. Our \texttt{BPR+Ours}\footnote{To set $\lambda$ in our approach, we assumed to work in a context where the platform owners are willing to lose 10\% of NDCG at most to decrease as much as possible the disparate exposure. We hence found the highest $\lambda$ able to meet this constraint. Our preliminary experiments showed that the sensitiveness to this parameter depends on the dataset.} was compared against a range of methods aimed to balance between the two objectives, namely utility and disparate provide group representation:
\begin{itemize}[leftmargin=4mm]
\item \texttt{BPR+LFRank} \cite{DBLP:conf/ssdbm/YangS17}: a re-ranking approach that aims to learn a mapping that satisfies statistical parity while preserving utility. 
\item \texttt{BPR+FA*IR} \cite{DBLP:conf/cikm/ZehlikeB0HMB17}: a re-ranking approach that aims to ensure that the proportion of protected items in every prefix of the top-$k$ ranking remains above a given minimum.
\item \texttt{BPR+FOEIR} \cite{DBLP:conf/kdd/SinghJ18}: a re-ranking approach that aims to maximize the utility for the user while satisfying a notion of fairness. 
\item \texttt{BPR+GDE} \cite{DBLP:conf/ecir/GomezBS21}: a re-ranking approach that aims to mitigate position bias, to an iterative algorithm.
\end{itemize}

\vspace{2mm} \noindent Preliminary experiments showed that the \texttt{Par} policy appears as the most challenging one. Therefore, for the sake of space, we compare against baselines only under this policy. 

\begin{table}[!b]
\vspace{4mm}
\caption{Recommendation utility ($NDCG$, the higher it is, the better), minority group exposure ($\Tilde{E}_m$; it should be as closer as possible to the target exposure $E_m$ reported between parenthesis under the Policy column), and beyond-accuracy objectives, namely coverage, diversity, novelty (the higher they are, the better) on top-10 recommended lists.}
\label{tab:table-1}
\resizebox{1.0\linewidth}{!}{
\begin{tabular}{lll|ll|lll}
\textbf{Data}                            & \textbf{Approach} & \textbf{Policy}             & \textbf{NDCG} $\uparrow$      & \textbf{$\Tilde{E}_m$} & \textbf{Cov} $\uparrow$ & \textbf{Div} $\uparrow$ & \textbf{Nov} $\uparrow$ \\
\toprule
\multirow{9}{*}{\textbf{ML-1M}} & \texttt{BPR} & -         &   \textbf{0.13}    &  0.06   &  \textbf{0.53}   &   0.22  &   0.09  \\
& \texttt{BPR+Ours} & \texttt{Cat ($E_m=0.10$)} &    \underline{0.12}    &  0.09  &  0.50   &  \underline{0.23}   &   \underline{0.11}  \\
                                &                                    & \texttt{Int ($E_m=0.07$)} &    \textbf{0.13} &  0.07   &  \underline{0.52}   &  0.22   &  0.10   \\
                                &                                    & \texttt{Per ($E_m=0.07$)} &    \textbf{0.13}    &  0.07  &   \underline{0.52}  &   \textbf{0.28} *  &  \textbf{0.12} *  \\
                                \cline{4-8}
                                &                                    & \texttt{Par ($E_m=0.50$)} &    \textbf{0.11} *    &  \textbf{0.47}  &  \underline{0.49}   &  \textbf{0.26} *  &  \textbf{0.13} *   \\
                                & \texttt{BPR+LFRank} & \texttt{Par ($E_m=0.50$)}  &   0.09 *    &  0.39  &   \textbf{0.50}  & 0.21    &   0.10  \\
                                & \texttt{BPR+FA*IR} & \texttt{Par ($E_m=0.50$)}   &   0.08 *    &  0.40  & \textbf{0.50}    & \underline{0.25} *    &   \underline{0.11}  \\
                                & \texttt{BPR+FOEIR} & \texttt{Par ($E_m=0.50$)}  &   0.07 *    &  \underline{0.55}  & 0.48    &  \textbf{0.26} *  &  \textbf{0.13} *   \\
                                & \texttt{BPR+GDE} & \texttt{Par ($E_m=0.50$)}     &  \underline{0.10} *    &  0.36  &  0.49   & 0.22     &   0.09  \\
                                \midrule
\multirow{9}{*}{\textbf{COCO}}  & \texttt{BPR} & -         & \textbf{0.05} &  0.08 &  \textbf{0.19}   &  0.25   &  0.06   \\
& \texttt{BPR+Ours} & \texttt{Cat ($E_m=0.19$)} & \underline{0.04} &  0.17  &  \underline{0.18}   &  \underline{0.27} *  &    \underline{0.09} \\
                                &                                    & \texttt{Int ($E_m=0.12$)} & \textbf{0.05} & 0.10    &  \textbf{0.19}   &  0.26   &  \underline{0.09}   \\
                                &                                    & \texttt{Per ($E_m=0.12$)} & \textbf{0.05} &  0.11  &  \textbf{0.19}   &  \textbf{0.40} *   &  \textbf{0.19} *   \\
                                \cline{4-8}
                                &                                    & \texttt{Par ($E_m=0.50$)} & \underline{0.04} &  \textbf{0.53} &  \underline{0.16}   &  \textbf{0.27} *   &  \underline{0.14} *   \\
                                & \texttt{BPR+LFRank} & \texttt{Par ($E_m=0.50$)}  & \textbf{0.05} &  0.41  &  \underline{0.16}   &  \underline{0.26}   &   0.10 *  \\
                                & \texttt{BPR+FA*IR} & \texttt{Par ($E_m=0.50$)}   & 0.02 * &  0.43  &  \textbf{0.17}   &  0.24   &   0.10 *  \\
                                & \texttt{BPR+FOEIR} & \texttt{Par ($E_m=0.50$)}  & 0.02 * & 0.56    &  0.15   &  \textbf{0.27} *   &   0.13 *  \\
                                & \texttt{BPR+GDE} & \texttt{Par ($E_m=0.50$)}     & \underline{0.04} &  \underline{0.55}  &  \underline{0.16}   &   \textbf{0.27} * &   \textbf{0.15} *  \\
                                
                                \bottomrule
    \multicolumn{8}{l}{For each data set: best result in \textbf{bold}, second-best result \underline{underlined}.}\tabularnewline
    \multicolumn{8}{l}{Statistical significance level, w.r.t. the original model \texttt{BPR}: (*) $p \le 0.05$.}\tabularnewline
\end{tabular}}
\end{table}

\vspace{2mm}  \noindent \textbf{Results and Discussion}. 
In a first analysis, we investigated whether our re-ranking approach can achieve a better trade-off between recommendation utility and disparate provider group exposure compared to state-of-the-art baselines. \tablename~\ref{tab:table-1} collects the Normalized Discounted Cumulative Gain ($NDCG$) \cite{WangWLHL13} (binary relevance scores and a base-2 logarithm decay; the higher the score is, the higher the recommendation utility) and the exposure achieved in the recommendations for the minority group ($\Tilde{E}_m$). We computed Kruskal-Wallis tests to establish the statistical significance of the differences between methods and then a series of post-hoc Dunn’s tests to establish significance at the pairwise level.

It can be observed that not all policies led to a loss in $NDCG$. In ML-1M, under the \texttt{Int} and the \texttt{Per} policies, no difference was measured with respect to the baseline \texttt{BPR}. Subsequently, \texttt{Cat} and \texttt{Par} showed a decrease ($\leq 0.02$ NDCG points) in $NDCG$. This pattern would have been expected, given that the latter two policies require applying a larger change in the provider group exposure distribution achieved by the original recommender system. For instance, for the minority group, its exposure in the original recommendations for \texttt{Cat} ($0.06$) should be moved up to $0.10$. This needed change was even higher for the \texttt{Par} policy, which aimed to reach an exposure of $0.50$ for both provider groups. Compared to the baselines, our approach \texttt{BPR+Ours} consistently achieve higher $NDCG$, even under the most challenging \texttt{Par} policy. Similar observations can be made also in COCO. In this dataset, \texttt{BPR+LFRank} and \texttt{BPR+GDE} achieved a comparable (slightly higher) $NDCG$ w.r.t. \texttt{BPR+Ours} under \texttt{Par}. 

Considering the exposure $\Tilde{E}_m$ given to the minority group (according to the target), it can be observed that our approach \texttt{BPR+Ours} was able to reach a group exposure closer to the target of a certain policy, with an often negligible loss in $NDCG$. The main advantage is that our approach has a fine-grained yet customizable control on the exact level of target exposure, that makes it more flexible especially in cases where the exposure distribution changes should be larger (\texttt{Cat} and \texttt{Par} policies). Under the \texttt{Par} policy, compared to the baselines, \texttt{BPR+Ours} led to an achieved minority group exposure closer to the target $0.5$, yielding also the highest $NDCG$ in ML-1M and the second highest $NDCG$ in COCO. In the latter dataset, however, compared to \texttt{BPR+LFRank}, our approach has a loss of $0.01$ in $NDCG$ in the face of $0.14$ of gain in $\Tilde{E}_m$. 

We can conclude that \texttt{BPR+Ours} attains a better trade-off between recommendation utility and disparate exposure compared to the considered baselines. Though the gains on disparate exposure are at the cost of slight decreases of recommendation utility, the observed loss is often negligible and not significant, making it easier to decide how to balance them (\textbf{\textsc{RQ1}}).

In a second analysis, we were interested in exploring whether our re-ranking approach provides benefits pertaining to relevant beyond-accuracy objectives for the platform in terms of coverage, diversity, and novelty. \tablename~\ref{tab:table-1} reports the Catalog Coverage ($Cov$), Category Diversity ($Div$), and Novelty ($Nov$) \cite{kaminskas2016diversity}. The higher the metric score is, the more the beyond-accuracy objective is met. 

We observe that, especially under \texttt{Par}, the catalog coverage was reduced (though \texttt{BPR+Ours} still has state-of-the-art performance compared to the baselines, in addition to a better utility-exposure trade-off). This loss in coverage might be due to the fact that all approaches tend to frequently remove items of the majority group that are rarely present in the top-10 list and recommend, in the top-10 list of more users, items of the minority group that were already covered by other users. Though this practice can lead to a better utility-exposure trade-off, the overall catalog coverage ends up being slightly reduced. Conversely, by better increasing the exposure of the minority group, \texttt{BPR+Ours} achieved better list-level beyond-accuracy objectives (category diversity and item novelty) w.r.t. the considered baselines, under \texttt{Par}. Comparing our approach across policies, \texttt{Per} led to the highest category diversity and item novelty, while the loss in catalog coverage was confirmed across all the other policies as well (though being smaller). 

We can conclude that \texttt{BPR+Ours} also better achieves beyond-accuracy objectives at a recommended list level (category diversity and item novelty), whereas it suffers from a small loss in beyond-accuracy objectives at a global level (catalog coverage) (\textbf{\textsc{RQ2}}). 

\section{Conclusions and Future Work}
\label{sec:conclusions}
In this paper, we proposed a re-ranking approach based on maximal marginal relevance that is able to minimize disparate exposure for provider groups according to a certain recommendation distribution policy. Based on the results, supporting all provider groups, especially minorities, led to a small loss in recommendation utility. The higher the difference between current and target provider group representations, the higher the loss. Furthermore, reducing inequalities in exposure among provider groups resulted in a positive impact on beyond-accuracy objectives. Overall, our approach led to a better trade-off among recommendation utility, exposure, and beyond-accuracy objectives, than the considered baselines. 

In the following steps, we plan to consider other datasets, recommender systems, domain-specific recommendation policies, and scenarios with (more than) two groups of providers (e.g., geographic- or age-based groups, unpopular provider groups). Thanks to our design choices and the resulting flexibility, our approach is fully agnostic concerning all these experimental elements. 

\balance
\bibliographystyle{ACM-Reference-Format}
\bibliography{sample-base}

\end{document}